\documentclass[11pt]{article}\usepackage{amsmath,amssymb,amsthm,amsfonts,mathptmx,mathrsfs,fullpage}
\usepackage{tablefootnote}
\usepackage[margin=1in]{geometry}
\usepackage{graphicx,color,cases}
\usepackage{hyperref}  
\usepackage{xfrac}  
\definecolor{DarkGreen}{rgb}{0.1,0.5,0.1}
\definecolor{DarkRed}{rgb}{0.5,0.1,0.1}
\definecolor{DarkBlue}{rgb}{0.1,0.1,0.5}
\hypersetup{
    unicode=false,              pdftoolbar=true,            pdfmenubar=true,            pdffitwindow=false,          pdfnewwindow=true,          colorlinks=true,           linkcolor=DarkBlue,              citecolor=DarkGreen,            filecolor=DarkGreen,          urlcolor=DarkBlue,                          pdftitle={Computing spectral bounds of the Heisenberg ferromagnet from geometric considerations},
    pdfauthor={Yingkai Ouyang},
        }

\usepackage{algorithm} 
\usepackage{algpseudocode,pifont,float}

 \def\case#1{{\left\{  
	\begin{array}{ll}
  #1
	\end{array}
 \right.   }}

\newtheorem{theorem}{Theorem}[section]

\newtheorem{lemma}[theorem]{Lemma}

\newtheorem{corollary}[theorem]{Corollary}
\newtheorem{conjecture}[theorem]{Conjecture}

\newtheorem{proposition}[theorem]{Proposition}
\theoremstyle{definition}

\newtheorem{definition}[theorem]{Definition}

\numberwithin{equation}{section}
\newtheorem{algo}[theorem]{Algorithm}

\newcommand{\sympow}[1]{{^{\{ #1\}}}}

\def\>{\rangle} 
\def\<{\langle}

\def\case#1{{\left\{  
	\begin{array}{ll}
  #1
	\end{array}
 \right.   }}

\floatname{algorithm}{Algorithm}

\def\>{\rangle} 
\def\<{\langle}

\newcommand{\floor}[1]{{\left \lfloor {#1} \right \rfloor  }}

\begin{document}

\title{Computing spectral bounds of the Heisenberg ferromagnet from geometric considerations} 
\author{Yingkai Ouyang  
\footnote{y.ouyang@sheffield.ac.uk
\newline 
} 
} 
\date{
\small University of Sheffield, Department of Physics and Astronomy, 226 Hounsfield Rd, Sheffield S3 7RH, UK\\
\small Singapore University of Technology and Design, 
8 Somapah Road, Singapore 487372.\\
Centre for Quantum Technologies, National University of Singapore,
3 Science Drive 2, Singapore 117543.\\
\today}

\maketitle

\begin{abstract}
We give a polynomial-time algorithm for computing upper bounds on some of the smaller energy eigenvalues in a spin-1/2 ferromagnetic Heisenberg model with any graph $G$ for the underlying interactions.
An important ingredient is the connection between Heisenberg models and the symmetric products of $G$. 
Our algorithms for computing upper bounds are based on generalized diameters of graphs. 
Computing the upper bounds amounts to solving the minimum assignment problem on $G$, which has well-known polynomial-time algorithms from the field of combinatorial optimization.
We also study the possibility of computing the lower bounds on some of the smaller energy eigenvalues of Heisenberg models.
This amounts to estimating the isoperimetric inequalities of the symmetric product of graphs. By using connections with discrete Sobolev inequalities, we show that this can be performed by considering just the vertex-induced subgraphs of $G$.
If our conjecture for a polynomial time approximation algorithm to solve the edge-isoperimetric problem holds,
then our proposed method of estimating the energy eigenvalues via approximating the edge-isoperimetric properties of vertex-induced subgraphs will yield a polynomial time algorithm for estimating the smaller energy eigenvalues of the Heisenberg ferromagnet.
\end{abstract}

  \section{Introduction}
  
The Heisenberg model (HM) is a quantum theory of magnetism \cite{Heisenberg1928},
and is prevalent in many naturally occurring physical systems 
such in various cuprates \cite{PhysRevLett.76.3212,chung2001large}, in solid Helium-3 \cite{thouless1965exchange}, and more generally in systems with interacting electrons \cite{Blundell}.
The HM can also be engineered in ultracold atomic gases  \cite{duan2003controlling} and quantum dots \cite{tamura2004tunable}.
Given the abundance of the HM, it may be advantageous to obtain a detailed understanding of its spectral structure. 
Such an understanding would for example help us to analyze the 
feasibility of storing quantum information in HMs via 
encoding into permutation-invariant quantum error correction codes \cite{Rus00,PoR04,ouyang2014permutation,ouyang2015permutation,OUYANG201743,ouyang2019quantum}.
Moreover, given the widespread applicability of magnetic material in classical information processing \cite{cullity2011introduction,jiles2015introduction},
quantum magnets based on the HM could similarly enable quantum technologies.
In addition, the HM also can be used for quantum computation \cite{divincenzo2000universal} and
quantum simulation.
What is most interesting is the relevance of the HM in mathematical physics because it is a paradigmatic model of statistical mechanics. For example, the celebrated Mermin-Wagner theorem \cite{merminwagner-PhysRevLett.17.1133} was proven for the HM.
  
The central object in this paper is the Heisenberg Hamiltonian (HH). 
It is the mathematical embodiment of the HM's energy level structure,
and contains all information necessary to derive every property of the HM.
More precisely, the HH for spin-half particles in the absence of an external magnetic field is a matrix given by
\begin{align}
 \hat H  
= -  \sum_{ \{i,j\} }J_{\{i,j\}}  
\frac{ 
 \sigma^x_i \sigma^x_j +
  \sigma^y_i \sigma^y_j +
   \sigma^z_i \sigma^z_j -  {\bf 1}}{2}, \label{eq1}
\end{align}
where ${\bf 1}$ is the identity matrix, $\sigma^x_i, \sigma^y_i $ and $\sigma^z_i $ as the usual Pauli matrices acting on the $i$-th particle, the sets $\{i,j\}$ are included in the sum whenever particles $i$ and $j$ interact, and $J_{\{i,j\}}$ is an exchange constant which quantifies the strength and nature of the coupling between the particles.
Here, we restrict our attention to ferromagnetic HHs, where every exchange constant is non-negative. 
We write the Hamiltonian in this way because we want the smallest eigenvalue of $\hat H$ to be zero.  
It is a well-known fact that a ferromagnetic HH can be written as a sum of graph Laplacians.
For completeness, we give its proof later in Theorem \ref{theorem:decomposition}.
Studying the spectrum of the HH is thus equivalent to studying the spectrum of these Laplacians.
The field of spectral graph theory deals entirely on determining the eigenvalues of graph Laplacians, 
and there has been an extensive amount of work done on this topic. 
One can for example refer to Chung's book for a review of the most important results in spectral graph theory \cite{chung1997spectral}.

Traditionally, most studies on the HM rely on the Bethe ansatz \cite{Bet31}.
In such approaches, the structure of the eigenvectors is assumed, and later verified to hold by solving for some of the previously undetermined parameters. 
This approach has proved hugely successful in 1D Heisenberg models \cite{haldane1983,Fad84,Koma87,EAT94,Kennedy1985,Kennedy1990,PhysRevLett.76.3212,Ogata2016}.
Recently, lower bounds have been proved on the average free energy of the HM on three dimensional lattices \cite{correggi2014validity} and also on lattices with any dimension \cite{correggi2015validity}.
However bounds on the spectrum of the Heisenberg ferromagnet have yet to be directly addressed.
Moreover, while certain other 2D HMs have been studied \cite{shastry1981exact,Sha88,baker1967two,chung2001large},
the question of how to address HMs of potentially arbitrary geometry remains unresolved.

In recent years, there has been impressive progress towards determining the spectrum of the HH.
The seminal result of Caputo, Liggett and Richthammer proves the Aldous' spectral gap conjecture \cite{caputo2010proof},
which implies that the spectral gap of the HH is equal to the spectral gap of Laplacian representing the graph of interactions of the HM. Since the size of this Laplacian is just the number of the HM's spins, determining the spectral gap of the HH is completely trivial, and can be found numerically in polynomial time \cite{spielman2014nearly}.
One of the most important developments thereafter was made by Correggi, Giuliani and Seiringer \cite{correggi2014validity, correggi2015validity}
where they develop important Sobolev inequalities for discrete graphs, but which are also applicable to the HM.  
Based on this, they find the right inequalities to obtain lower bounds on the free energy of the HM at finite temperatures.
However the problem of obtaining bounds for the higher eigenvalues of the HH has been largely unaddressed.

In this paper, we utilize relatively recent developments in spectral graph theory to obtain new bounds for HH's spectrum.
With regards to the upper bounds, we rely on analytical bounds on the eigenvalues of a graph based on its generalized diameters by Chung, Grigoryan and Yau \cite{CGY96}.
For the lower bounds we use Chung and Yau's Sobolev inequalities on graphs \cite{ChY95}.
There are two innovations provided in this paper. 
First, we identify a probabilistic polynomial-time algorithm to obtain upper bounds on the HH's eigenvalues by reducing the computation of a generalized diameter to that of a minimum assignment problem.
Second, we provide new discrete Sobolev inequalities that are based on deleting vertices from graphs. These inequalities can be used with Chung and Yau's Sobolev inequalities to obtain lower bounds on the eigenvalues of the HH.
To the best of our knowledge, this is the first time graph-theoretic methods are directly used to obtain bounds on the eigenvalues of the HH.

We begin our paper by explaining how the HH is connected to the symmetric power of graphs in Section \ref{sec:graphs-intro}. 
In a preprint by Rudolph, the connection between graphs and the HH was noted, and the terminology
of symmetric power of graphs was coined \cite{Rud02}.
Such graphs, later also known as token graphs \cite{fabila2012token}, 
have been extensively studied in recent years for their graph theoretic properties in \cite{AGRR07,alzaga2010spectra,yamanaka2015swapping,leanos2018connectivity} among many others.
Once we establish the connection of the HH with symmetric powers of graphs, 
we turn our attention to the elementary problem of determining the spectrum of the mean-field Heisenberg ferromagnet, where every pair of spins interacts with the same exchange constant.
Obviously, the SU(2) symmetry of such a model immediately allows one to determine the HH eigenvalues and multiplicities, 
and the eigenprojectors can be in principle calculated using textbook methods with Clebsch-Gordan coefficients.
However we wish to highlight that by using well-known facts about association schemes, 
we can already directly identify the eigenprojectors of this HH in terms of Hahn polynomials 
and generalized adjacency matrices (see Theorem \ref{theorem:mean-field-SD}).

Generalized diameters of graphs play a central role in deriving upper bounds on the spectrum of HHs, as we shall see in Section \ref{sec:upper}.
These generalized diameters can be thought of as the widths of a body when it is interpreted to have a given dimension.
The most important feature of our algorithms is that they run much more efficiently than algorithms that attempt to directly evaluate the eigenvalues of the HH. 
We show that computing these generalized diameters is equivalent to the minimum assignment problem, which is solved efficiently using the Kuhn-Munkres algorithm \cite[Page 52]{schrijver2004combinatorial}. 
Together with analytical bounds on the eigenvalues of a graph based on its generalized diameters by Chung, Grigoryan and Yau \cite{CGY96}, 
we thereby obtain a polynomial-time algorithm for evaluating upper bounds on the eigenvalues of the ferromagnetic HH,
which gives us our result in Theorem \ref{theorem:upper-bound-algo}. 

Isoperimetric inequalities play a central role in deriving lower bounds on the spectrum of HHs in this paper.
An isoperimetric inequality essentially gives a lower bound on the the minimum boundary size of a body with a fixed volume in a given manifold. 
Specializing this to graphs, we require a lower bound on the minimum cut-size of a subset of $k$ vertices, for every possible choice of $k$. 
Such bounds are then called edge-isoperimetric inequalities, which we introduce in Section \ref{sec:computing-lowerbounds}.
Based on edge-isoperimetric inequalities of the symmetric products of graphs, we present lower bounds on the eigenvalues of the ferromagnetic HH (see Theorem \ref{theorem:lowerbound}).
Because deriving edge-isoperimetric inequalities on the symmetric product of graphs is potentially difficult, 
we also derive isoperimetric inequalities on the symmetric product of graphs based on the isoperimetric inequalities on their vertex-induced subgraphs (see Theorem \ref{thm:symprod-seminorm}). 
We introduce some Sobolev inequalities in Section \ref{sec:sobolev}, and proceed to use our results on isoperimetric inequalities on the symmetric product of graphs to obtain lower bounds on all of the HH eigenvalues based on isoperimetric properties of the associated graphs.  
For this, we use the Sobolev inequalties of with Chung, Yau \cite{ChY95} and Ostrovskii's \cite{Ost05} on graphs.

Finally in Section \ref{sec:conclusion}, we discuss some potential implications of our bounds and algorithms.
We then remark on the potential to improve both the upper and lower bounds that we present, 
by further investigation using a combinatorial approach.
We also point out how an advance in the field of approximation algorithms could help to 
make computing lower bounds for the spectrum of the ferromagnetic HH much more efficient. 

\section{Graphs and the Heisenberg model}
\label{sec:graphs-intro}

Since we investigate the spectrum of HHs with graphs of varying dimensions, we need to explain what these graphs and their dimensions are.
Here, a graph corresponding to a HH comprises of vertices from 1 to $n$ which label the particles,
and edges $\{u,v\}$ which label the interaction between particles $u$ and $v$.
A graph's dimension generalizes from the dimension of continuous manifolds.
The edge-boundary of any set of vertices $X$ denoted by $\partial X$ is 
the set of edges in $G$ with exactly one vertex in $X$.  
Suppose that every set $X$ with $k$ vertices in $G$
satisfies the bound $|\partial X| \ge c k^{1-1/\delta}$ for some positive constant $c$ for every $k\le n/2$. 
Then we say that $G$ has a dimension of $\delta$ with isoperimetric number $c$.
This is analogous to the situation where a manifold with fixed volume $k$ and a surface area of at least $c k^{1-1/\delta}$ for some positive constant $c$ has a dimension of $\delta$.
The dimension of a physical system is then the dimension of the corresponding graph of interactions.

To understand how precisely HH is related to graphs, we need to define the symmetric product of a graph.
When $k$ is a non-negative integer with $k \le n$,
the $k$-th symmetric product of a graph $G$ with vertices $V$ and edges $E$ denoted by $G \sympow k$ is a graph with the following properties.
First, $G \sympow k$ has as its vertices all possible subsets of $V$ of size $k$.
Second, the edges of $G \sympow k$ are the sets $\{X, Y\}$ where 
(i) $X$ and $Y$ are subsets of $V$ with $k$ vertices, 
(ii) $X$ and $Y$ have $k-1$ common elements,
and (iii) their symmetric difference, the union of the sets without their intersection, is an edge in $E$. 
In short, $\{X, Y\}$ is an edge in $G \sympow k$ only if the symmetric difference of $X$ and $Y$ is an edge in $E$, {\em i.e} $X \triangle Y \in E$. Examples of the symmetric product of graphs can be seen in Figure \ref{fig:G1} and Figure \ref{fig:G3}.
 \begin{figure}
  \includegraphics[width=\textwidth]{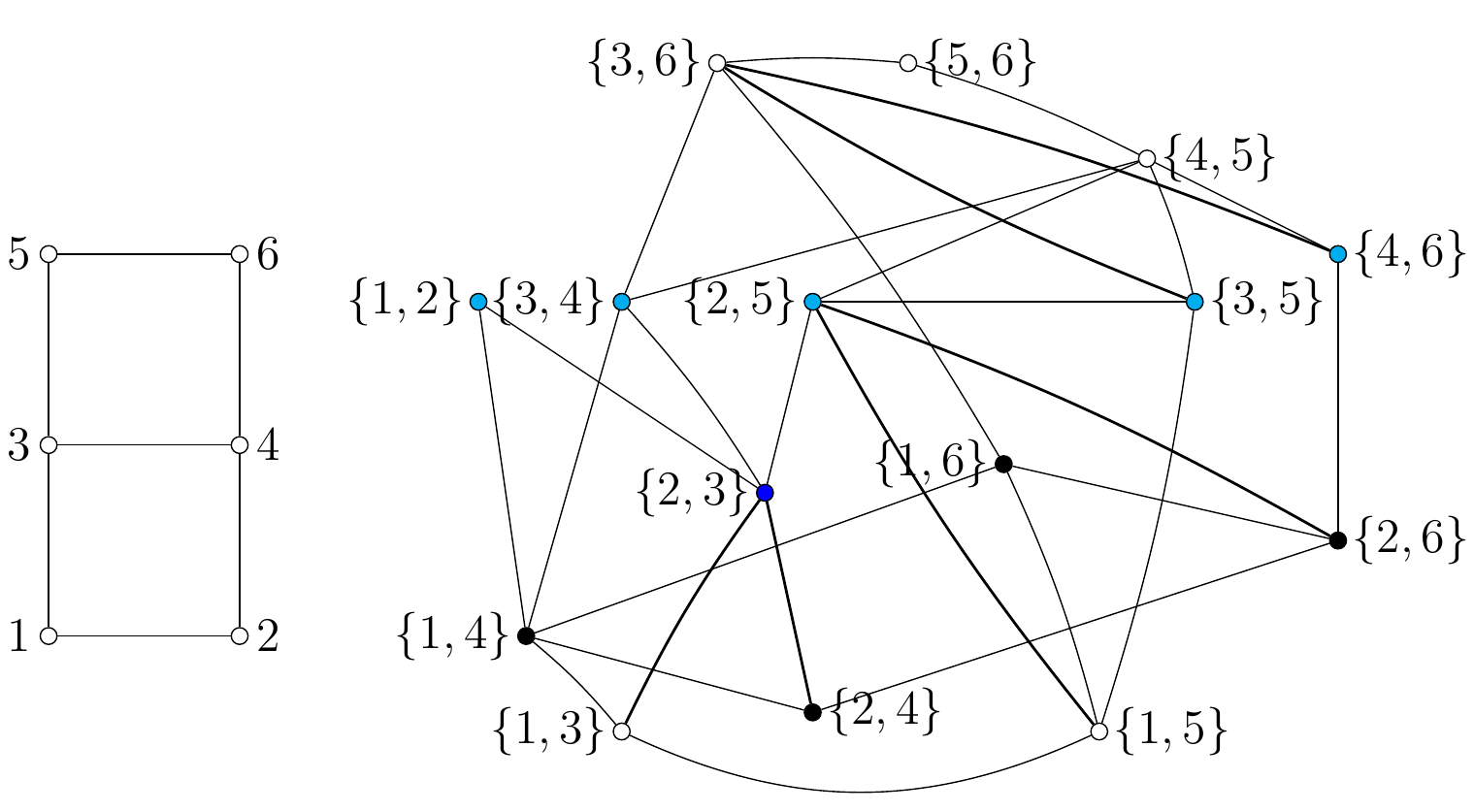}
  \caption{
  On the left is a graph $G$ with six vertices, and on the right is its symmetric square $G \sympow 2$.
  The symmetric cube $G \sympow 3$ is depicted in Fig.~\ref{fig:G3}.
  }
  \label{fig:G1}
\end{figure}

\begin{figure}
  \includegraphics[width=\textwidth]{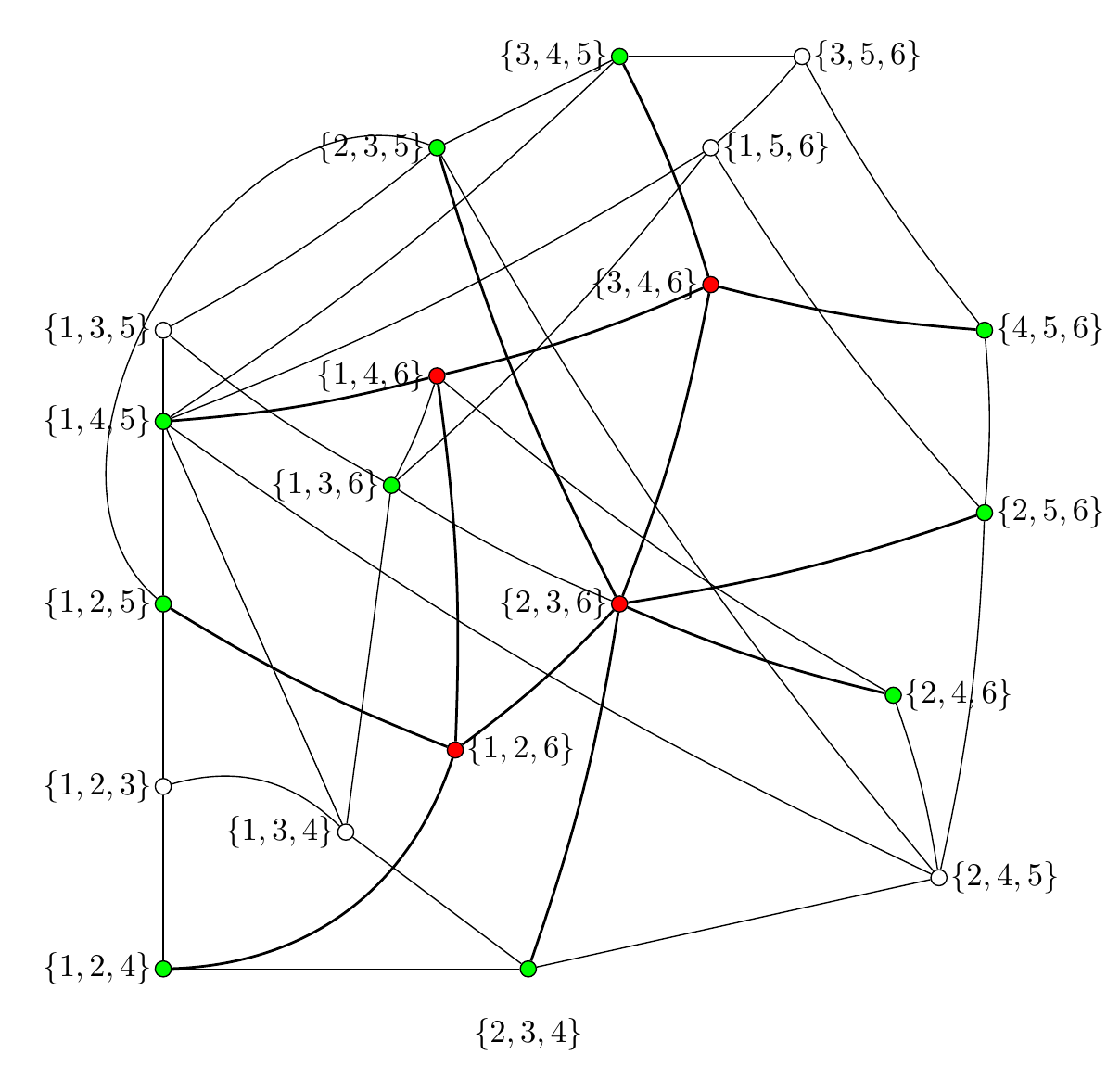}
  \caption{$G \sympow 3$, the symmetric cube of the graph $G$ depicted in Figure \ref{fig:G1}, is shown here. }
  \label{fig:G3}
\end{figure}

Now we proceed to define the Laplacians of $G \sympow k$.
By denoting $|X\>$ as a state with the spins labeled by $X$ in the up state and the remaining spins in the down state where $X$ is a subset of vertices in $G$, 
the Laplacians of $G \sympow k$ are
\begin{align}
L_k \!
=\!\!\!\!
\sum_{\substack{ X \subseteq V \\ |X| = k } }
\!\!\!\!
|\partial X| | X \>\< X | 
-
\sum_{ \substack { X \triangle Y \in E   \\ } }   
( |X\>\<Y| + |Y\>\<X|) \label{eq:Lk-defi}.
\end{align}  
Here, each $L_k$ is the Laplacian of the graph $G \sympow k$ and has rank $\binom n k$.
If we interpret $G \sympow k$ as a discrete manifold, 
the eigenvectors and eigenvalues of $L_k$ are its normal modes and associated resonance frequencies.

If we normalize the HH so that every non-zero exchange constant is equal to 1, we get the normalized Hamiltonian 
\begin{align}
 \hat H_1
= -  \sum_{ \{i,j\} \in E } 
\frac{ 
 \sigma^x_i \sigma^x_j +
  \sigma^y_i \sigma^y_j +
   \sigma^z_i \sigma^z_j -  {\bf 1}}{2}. \label{eq:H1-defi}
\end{align}
This normalized Hamiltonian $\hat H_1$ is just a sum of pairwise orthogonal matrices $L_k$ \cite[Appendix A]{AGRR07},
as we can see from the following theorem.
\begin{theorem}
\label{theorem:decomposition}
Let $G=(V,E)$ be a graph with $n$ vertices.
Then $\hat H_1 = L_0 + \dots + L_n$ where $L_k$ are as given in Eq.~(\ref{eq:Lk-defi}) and $\hat H_1$ is as given in Eq.~(\ref{eq:H1-defi}).
\end{theorem} 
This decomposition of the ferromagnetic HH with graph $G$ as sum of pairwise orthogonal matrices, with each matrix associated with the symmetric products of $G$, has already been known for years \cite[Appendix A]{AGRR07}. 

The decomposition of the normalized Hamiltonian as given in Theorem 
\ref{theorem:decomposition} holds because of its fundamental connections with Laplacians in graph theory \cite[Chapter 13]{godsil2001algebraic}. 
Using a graph-theoretic perspective, some trivial properties this normalized Hamiltonian can be easily seen.
For example, when the graph $G$ is connected, 
each $L_k$ has exactly one eigenvalue equal to zero with corresponding eigenvector 
$\binom {n}{k}^{-1/2}\sum_{\substack{X \subseteq V\\ |X| = k}}|X\>$  \cite[Lemma 13.1.1]{godsil2001algebraic}. 
Hence the ground state energy of $\hat H_1$ is zero with degeneracy $n+1$,
and the ground space is spanned by the Dicke states $|D^n_k\>$ \cite{ouyang2014permutation}, 
where $|D^n_k\>$ is a normalized superposition of all $|X\>$ for which $X$ is a subset of $\{1,\dots, n\}$ of size $k$.
Moreover, for any graph, the Laplacians $L_k$ and $L_{n-k}$ are unitarily equivalent,
because of the equivalence of $G \sympow k$ and $G \sympow {n-k}$ under set complementation.
To see this, denote $\overline X$ as the set complement of $X \subset V$, 
and note that $L_{n-k} =  U_k L_k U_k ^\dagger$ where
\begin{align}
U_k = \sum_{X \subseteq V : |X|=k } |\overline X\> \< X|  	. \label{eq:U-definition}
\end{align}
Hence it suffices to only study Laplacians $L_k$ for which $k \le \frac n 2$. 

The implication of Caputo, Liggett and Richthammer's proof of Aldous' spectral gap conjecture  \cite{caputo2010proof} is that the spectral gap of every $L_k$ for $k=1,\dots, n-1$ is identical.
This renders the problem of finding the spectral gap of HHs trivial, because $L_1$ is effectively a size $n$ matrix and its spectral gap can be efficiently solved numerically, 
for example by using Spielman and Teng's celebrated algorithm \cite{spielman2014nearly}.

In this paper, we will focus on the obtaining bounds of the eigenvalues of every $L_k$, which we denote as 
$\lambda_0(L_k), \lambda_1(L_k), \dots, \lambda_{\binom n k -1}(L_k)$. We call $\lambda_1(L_k)$ the spectral gap of $L_k$ and $ \lambda_{\rm max}(L_k)=  \lambda_{\binom n k -1}(L_k)$ the largest eigenvalue of $L_k$. We order these eigenvalues so that
\begin{align}
0=\lambda_0(L_k)\le \dots \le \lambda_{\binom n k -1}(L_k).
\end{align}
Now we proceed to give the proof of 
Theorem \ref{theorem:decomposition}.
\begin{proof}[Proof of Theorem \ref{theorem:decomposition}] 
The first step is to notice that the swap operator of two qubits can be written as 
\begin{align}
& (|0\>\otimes |0\>)( \<0|\otimes \<0|)
 +
  (|0\>\otimes |1\>)( \<1|\otimes \<0|) \notag\\
 +&
  (|1\>\otimes |0\>)( \<0|\otimes \<1|)
 +
  (|1\>\otimes |1\>)( \<1|\otimes \<1|),
\end{align}
and is identical to the sum 
$\frac{\sigma^x_1 \sigma^x_2 + \sigma^y_1 \sigma^y_2  +\sigma^z_1 \sigma^z_2 + {\bf 1}}{2}.$
Then denoting the operator that swaps qubits $i$ and $j$ as $\pi_{i,j}$, we have the identity
\begin{align}
\pi_{i,j}- {\bf 1} =\frac{\sigma^x_i \sigma^x_j + \sigma^y_i \sigma^y_j  +\sigma^z_i \sigma^z_j - {\bf 1}}{2}.
\end{align}
This allows us to rewrite the normalized HH with a graph $G=(V,E)$ in terms of swap operators, so that
\begin{align}
\hat H_1 = \sum_{\{i,j\} \in E}  ({\bf 1} -\pi_{i,j}).
\end{align}
Next, we let $X$ denote any subset of vertices $V=\{1,\dots,n\}$.
Then for any distinct $i$ and $j$ from the set $V$, 
we have 
\begin{align}
\pi_{i,j} |X\> = \case{
|X\> &, i,j \in X \\
|X\> &, i,j \notin X \\
|X \triangle \{i,j\} \> &, \{i,j\} \in \partial X \\
}.
\end{align}
This allows us to obtain 
\begin{align}
\sum_{\{i,j\} \in E}\pi_{i,j} |X\>
&= 
\sum_{\{i,j\} \in \partial X}\pi_{i,j} |X\>
+
\sum_{\{i,j\} \notin \partial X}\pi_{i,j} |X\> \notag\\
&=
\sum_{\{i,j\} \in \partial X} |X \triangle \{i,j\}\>
+
\sum_{\{i,j\} \notin \partial X}|X\> \notag\\
&=
\sum_{\{i,j\} \in \partial X} |X \triangle \{i,j\}\>
+
(m-|\partial X|) |X\>,
\end{align}
where $m$ denotes the number of edges in $E$.
Hence 
\begin{align}
\hat H_1 |X\> 
&= 
\sum_{\{i,j\}\in E} |X\> -
\sum_{\{i,j\}\in E} \pi_{i,j} |X\>   \notag\\ 
&= 
|\partial X|  |X\> -
\sum_{ \{i,j\} \in \partial X}  |X \triangle \{i,j\}\>   .
\end{align}
Clearly if $Y$ is a subset of $V$ that has a different size from $X$, then 
$\<Y|\hat H_1 |X\> = 0$. This immediately implies that $\hat H_1$ can be written as a sum of orthogonal matrices, each of them supported on the space spanned by $|X\>$ where $X$ have constant size.
Next, note that $\<X|\hat H_1 |X\> = |\partial X|$, which implies that the diagonal entries of $L_k$ are given by the sizes of the corresponding edge-boundaries of $k$-sets.
Finally, note that if $Y$ has the same size as $X$, then $\<Y|\hat H_1 |X\> = 0$ whenever $X \triangle Y \notin E$ and  $\<Y|\hat H_1 |X\> = 0$ whenever $X \triangle Y \in E$. This proves the result.
\end{proof}

\section{Exact solutions for the mean-field model}

 We begin with a combinatorial approach for producing the exact solution for a mean-field HM. Such a HM has $n$ spins, 
 and every pair of spin interacts with exactly the same exchange constant $J$. 
 In this case, the normalized Hamiltonian is 
\begin{align}
\hat H_{1} = -  \sum_{i=1}^n \sum_{j=1}^{i-1} 
\frac{ \sigma^x_i \sigma^x_j +
  \sigma^y_i \sigma^y_j +
   \sigma^z_i \sigma^z_j  - {\bf 1}
 }{2}.
\end{align}
From the perspective of SU(2) symmetry, this model is trivial.
This is because we can write $\hat H_1 = - {\vec S_{\rm tot} \cdot \vec S_{\rm tot} \over 2} + \frac{n(n+1)}{2} {\bf 1} $,
where $\vec S_{\rm tot}=\sum_{i=1}^n \vec S_i$ and $\vec S_i = \vec\sigma_i/2
$. 
The spectrum along with the degeneracies is directly given by the representations contained
in the direct product of $n$ spin 1/2 representations,
\[1/2 \otimes 1/2 \otimes ... \otimes 1/2,\]
which can be easily solved using standard techniques.
Moreover, the corresponding eigenvectors can be in principle calculated using textbook methods with Clebsch-Gordan coefficients.
However this computation can be fairly tedious. 
We show how the eigenvalues and eigenprojectors of $\hat H_1$ can be alternatively obtained from a combinatorial perspective.

Note that for $\hat H_1$, the graph of interactions is precisely the complete graph on $n$ vertices.
The symmetric products of the complete graph are the Johnson graphs for which the spectral problem has been exactly solved using association schemes \cite{delsarte1973algebraic,bannai1984algebraic}.
Using this connection, we can use prior knowledge of the Johnson schemes to conclude that $L_k$ has exactly one eigenvalue equal to zero, and its other eigenvalues are $j(n+1-j)$ with multiplicities $m_j= \binom n j - \binom n {j-1}$ for $j = 1,\dots, k$ \cite[Section 12.3.2]{brouwer2011spectra}.
Hence the positive eigenvalues of $\hat H_{1}$
are 
\begin{align}
 j(n+1-j) 
\end{align}
with multiplicities 
\begin{align}
(n+1-2j)m_j,
\end{align}
where $j = 1,\dots, \floor{n/2}$.

What is most remarkable about the connection between association schemes and the mean-field Heisenberg model is that we can assign a combinatorial interpretation to the matrices $L_k$. In particular,  we can analytically decompose $L_k$ as a linear combination of eigenprojectors, where each eigenprojector is in turn a linear combination of generalized adjacency matrices.
We proceed to explain what these generalized adjacency matrices are.
Now the adjacency matrix of $L_k$ is 
\begin{align}
A_{k,1} = \sum_{ \substack { |X \triangle Y|=2   \\ } }   
( |X\>\<Y| + |Y\>\<X|).
\end{align}
Namely, the matrix element of $A_{k,1}$ labeled by $|X\>\<Y|$ has a coefficient of 1 if $X$ is adjacent to $Y$ in $G \sympow k$, and equal to zero otherwise.
Since two vertices in a graph are adjacent if and only if they are a distance of one apart, we can define the generalized adjacency matrices by having 
\begin{align}
A_{k,z} = \sum_{\substack{X ,Y\subseteq \{1,\dots, n\}\\ |X \triangle Y| = 2 z }} |X\>\<Y|.
\end{align}
Here, the matrix element of $A_{k,z}$ labeled by $|X\>\<Y|$ has a coefficient of 1 if $X$ is a distance of $z$ from $Y$ in $G \sympow k$, and equal to zero otherwise.
We call $A_{k,z}$ the $z$-th generalized adjacency matrix of the Johnson graph associated with $L_k$ relating $k$-sets a distance of $z$ apart.
For completeness, let $A_{k,0}$ denote a size $\binom n k$ identity matrix.
Now let 
\begin{align}
h_{k,j}(z) = m_j 
\sum_{a=0}^j (-1)^a 
\frac{   \binom{j}{a}    \binom{n+1-j}{a}    }
{    \binom{k}{a}    \binom{n-k}{a}   } \binom z a
\end{align}
denote a Hahn polynomial \cite[(18) and (20)]{delsarte1998association}.
Then, properties of the Johnson scheme given in Ref.~\cite{delsarte1998association} imply that for $k=1,\dots, \floor{n/2}$, the Laplacians $L_k$ have the spectral decomposition 
\begin{align}
L_k = \sum_{j=1}^k j(n+1-j) P_{k,j} \label{eq:lk-decomposition}
\end{align}
where 
\begin{align}
  P_{k,j} = \frac{1}{\binom n k  }
  \sum_{z=0}^k h_{k,j}(z) A_{k,z}
\end{align}
are pairwise orthogonal projectors. 
To make the spectral decomposition of the normalized mean-field HH explicit, we present the following theorem.

\begin{theorem}
\label{theorem:mean-field-SD}
Let $G=(V,E)$ be a complete graph. Then a normalized HH on this graph $\hat H_1$ has the spectral decomposition
\begin{align}
\hat H_{1}
&=  
\sum_{j=1}^{(n-1)/2} j(n+1-j) 
 \sum_{k=j}^{(n-1)/2}  
 \left(
P_{k,j} + U_k P_{k,j} U_k ^\dagger  
\right)
\end{align}  
when $n$ is odd, and 
\begin{align}
\hat H_{1}
&= 
\sum_{j=1}^{n/2-1} j(n+1-j) 
\left(
 \sum_{k=j}^{n/2}  
P_{k,j}
+
 \sum_{k=j}^{n/2-1}  
U_k P_{k,j} U_k ^\dagger  
\right)
+
\frac{n}{2}\left(
\frac{n}{2}+1
\right) P_{n/2,n/2}
\end{align}  
when $n$ is even.
\end{theorem}
\begin{proof}
The proof of this theorem relies on the identity
\begin{align}
\sum_{u=0}^a \sum_{j=1}^u a_{u,j}
=
\sum_{u=1}^a \sum_{j=1}^u a_{u,j}
=
\sum_{j=1}^a \sum_{u=j}^a  a_{u,j} \label{eq:rearrangement}
\end{align}
which holds for all non-negative integers $a$, and any complex coefficients $a_{u,j}$.

When $n$ is odd, 
we can write
\begin{align}
\hat H_1 = 
\sum_{k=0}^{(n-1)/2}
\left(
L_k +  U L_k U ^\dagger
\right),
\end{align}
where $U$ is the unitary as defined in (\ref{eq:U-definition}) 
and $L_j$ is as given in (\ref{eq:lk-decomposition}).
Substituting the decomposition of $L_j$, we get
\begin{align}
\hat H_1 = 
\sum_{k=0}^{(n-1)/2}
\sum_{j=1}^k 
j(n+1-j) 
\left(
P_{k,j}  +  U P_{k,j} U ^\dagger
\right),
\end{align}
Applying (\ref{eq:rearrangement}) then yields the result for odd $n$.
When $n$ is even, we have 
\begin{align}
\hat H_1 = 
\sum_{k=0}^{n/2-1}
\left(
L_k +  U L_k U ^\dagger
\right)
+L_{n/2}.
\end{align}
By using the techniques used to prove the case for odd $n$,
we get
\begin{align}
\hat H_{1}
&= 
\sum_{j=1}^{n/2-1} j(n+1-j) 
 \sum_{k=j}^{n/2-1}  
 \left(
P_{k,j} + U_k P_{k,j} U_k ^\dagger  
\right)
+
L_{n/2}.
\end{align}  
Substituting the value of (\ref{eq:lk-decomposition}) for $L_{n/2}$,
we get the result.
\end{proof}
•

\section{Upper bounds for the Heisenberg spectrum}
\label{sec:computing-upperbounds}
\subsection{Simple two-sided bounds on the largest eigenvalue}

We obtain bounds on the largest eigenvalue of ferromagnetic HHs with graphs having dimension $\delta$ with isoperimetric number $c$, and maximum vertex degrees $\beta$. 
Note that obtaining bounds on the largest eigenvalue of the normalized HH $\hat H_1$, 
amounts to obtaining bounds on $\lambda_{\rm max}(L_k)$. 
Now the largest eigenvalue of the Laplacian of any graph is at least its maximum vertex degree \cite[Page 149, line 7]{Merris1994} and at most twice its maximum vertex degree from Gersgorin's circle theorem \cite{Ger31,varga-GCT}.
The upper bound can also slightly improved over Gersgorin's circle theorem to be at most the sum of the largest and the second largest vertex degrees \cite[(6)]{Merris1994}.
Thus,
\begin{align}
 c k^{1-1/\delta}\le \lambda_{\rm max}(L_k) \le 2 k \beta
\end{align} 
for $1 \le k \le n/2$. Since $\hat H_1= L_0 + \dots + L_{n}$, we get 
\begin{align}
	c  \floor{n/2}^{1-1/\delta} \le  \lambda_{\rm max}(\hat H_1)   \le  n \beta.
\end{align}

\subsection{Upper bounds from graph diameters}
\label{sec:upper}
In this subsection, we outline an algorithmic approach for finding upper bounds on the smaller eigenvalues of the HH.
This approach relies crucially on the generalizations of the diameter of a graph.
The diameter of a graph is the length of its shortest path, and intuitively measures the size of the graph. In the case when the graph has the geometry of a hypercube of dimension $d$, its diameter will be the length between the vertices of the hypercube that are furthest apart. 
The generalization of the diameter that we will consider allows us to quantify, in the case of the hypercube, the length of its sides. In particular, the $d$-diameter of a $d$-dimensional hypercube will be precisely the length of its side. Intuitively, the $d$-diameter of a body is its width when it is interpreted to have $d$ dimensions.
The generalized diameters are important because they can give upper bounds on the eigenvalues of a graph Laplacian \cite{chung1996upper,CGY96}.

The generalized diameter of a graph quantifies its sparsity.
It is then reasonable to expect that the larger the generalized diameter, the smaller the upper bound on the eigenvalues can be, since a sparse graph ought to have smaller eigenvalues than a highly connected graph. In the extreme case when a graph comprises of disconnected vertices, its generalized distances are all infinite, and every eigenvalue is equal is zero. Thus in this case, we would anticipate that the upper bound we get from the diameter is also equal to zero. 
This is indeed the case. 
When a graph has $j$ distinct connected components, by selecting $j$ vertices, one from each of these connected components, the corresponding generalized distance is infinite.
This then implies that the $c$th smallest eigenvalue of the corresponding graph Laplacian is at most zero.
Since it is known that a graph with $c$ distinct components has a graph Laplacian with exactly $c+1$ zero eigenvalues \cite[Lemma 13.1.1]{godsil2001algebraic}, in this sense, the bound of \cite[Corollary 4.4]{CGY96} can be said to be tight. 

To understand the generalized diameter of a graph, we need to review the concept of the distance amongst a subset of its vertices. 
Now, the distance between a pair of vertices $v_a$ and $v_b$ is the just the length of the shortest path connecting them, which we denote as $d(v_a,v_b)$.
This can be computed using Algorithm \ref{alg:dist}. 
\begin{algo}
{\texttt{Dist}$(G=(V,E))$, Compute pairwise distances in $G$. \label{alg:dist}}
\begin{algorithmic}
\State $D \gets$ size $n$ matrix of zeros
\ForAll {$v \in V$}
    	\State Perform BFS on $v$, obtaining a spanning tree $T$ rooted at $v$.
        \ForAll {$w \in V, w \neq v$ }
			\State 	$D(u,v) \gets $ distance of vertex $w$ to $v$ in $T$
        \EndFor
    \EndFor 
	\State \textbf{return} $D$
\end{algorithmic} 
\end{algo}
The distance between a set of vertices $K=\{v_1,\dots, v_k\}$ is then the minimum pairwise distance between distinct vertices $v_a$ and $v_b$, which we denote as
\begin{align}
d(K) = \min\{ d(v_a, v_b) : 1 \le a < b \le k \}.
\end{align}
The $j$-diameter of a graph $G=(V,E)$ has been defined \cite[Page 25, last equation]{CGY96} as the maximum distance of subsets $K$ with $(j+1)$ vertices, and we denote it as
\begin{align}
d_j(G) = \max\{d(K):K\subseteq V, |K|=j+1\}.
\end{align}

Now define $d_{j,k}$ to be the $j$-diameter of $G \sympow k$.
Whenever $d_{j,k} \ge 2$, we can obtain upper bounds on the eigenvalues of $L_k$ from graph-theoretic results of Ref.~\cite[Corollary 4.4]{CGY96}.
\begin{align} \label{eq:Lkupperbound}
 \lambda_j(L_k)
 \le 
  \lambda_{\rm max}(L_k)
 \left( 
 1 - 2 / \left(  1+ \binom n k^{1/(d_{j,k}-1)} \right)  
 \right).
\end{align}  Clearly $d_{j,k}$ decreases with increasing $j$, and thus our upper bounds on $\lambda_j(L_k)$ are increasing with $j$ as one would expect.
Now let us see how (\ref{eq:Lkupperbound}) can be tight. Let us consider a graph $G$ with $c$ connected components and consider $k=1$, so that $G \sympow 1 = G$. 
We claim that the $(c-1)$-diameter of $G$ is infinite. This is because we can pick a set of vertices, with one vertex from each connected component. Since none of these vertices are connected, their pairwise distance is always infinite. Using this value for the generalized diameter, the upper bound in (\ref{eq:Lkupperbound}) for $\lambda_{c-1}(L_1)$ becomes zero. Since we know from \cite[Lemma 13.1.1]{godsil2001algebraic} that $\lambda_{c-1}L(G)=0$, the upper bound in (\ref{eq:Lkupperbound}) is tight.

Since the $j$-diameter of $G \sympow k$ may be unwieldy to calculate directly,  we outline a polynomial time algorithm to obtain lower bounds on it.
At the heart of our algorithm is the fact that the distances between vertices in $G \sympow k$ can be computed using only information about the distances between vertices in $G$. 
This makes it possible to estimate the $j$-diameter of $G \sympow k$ solely by computing on the graph $G$.
Before diving into the specifics of our algorithm, we briefly outline its inner workings.
\begin{enumerate}
\item Pick any $j+1$ distinct vertices $X_1, \dots, X_{j+1}$ from $G \sympow k$. Note that each of these vertices are subsets of $V$, each with $k$ elements.
\begin{algo}
{$\texttt{SEL}_k(j,V)$, Select $j+1$ distinct vertices in $G \sympow k$ \label{alg:select}}.
\begin{algorithmic}[0] 
	\State $X_1 \gets$ a random $k$-vertex subset of $V$
    \State $c \gets 1$
    
	\While{$c \le j+1$}  
    	\State $Y \gets$ a random $k$-vertex subset of $V$
        \If{$Y \cup X_a \neq Y$ for all $a =1,\dots,c$}  
        	\State $X_{c+1} \gets Y$
            \State $c \gets c + 1$
        \EndIf
    \EndWhile 
	\State \textbf{return} $(X_1,\dots,X_{j+1})$
\end{algorithmic}
\end{algo}
\item Loop over all $a,b$ such that $1 \le a < b \le j+1$.
\item Compute $d(X_a,X_b)$.
\begin{algo}
{$\texttt{dist}(X,Y,D)$, Evaluates the distance between $X$ and $Y$ in $G \sympow k$ \label{alg:distXY}}.
\begin{algorithmic}[0] 
	\State $Z \gets X \cap Y$
    \State $a \gets |X|-|Z| \cap Y$    
    \State $X=\{x_1,\dots, x_a\} \gets X \setminus Z$
	\State $Y=\{x_1,\dots, y_a\} \gets Y \setminus Z$
    \State $C \gets $zeros($a$) \Comment initialize a size $a$ matrix
    \ForAll{$u =1,\dots, a$}
    	\ForAll{$v =1,\dots, a$}
    		\State $C(u,v) = D(x_u,x_v)$
        \EndFor
    \EndFor
    \State $d \gets$ output of Kuhn-Munkres algorithm on the cost matrix $C$
	\State \textbf{return} $d$
\end{algorithmic}
\end{algo}
\item Exit loop.
\item A lower bound for  $d_j(G \sympow k)$ is the minimum $d(X_a,X_b)$.
\end{enumerate}
This procedure can in principle be repeated for all possible choices of  $X_1, \dots, X_{j+1}$ to obtain the value of $d_j(G \sympow k)$ exactly. Since this may be computationally expensive, we propose just to randomly select the vertices $X_1, \dots, X_{j+1}$ a constant number of times.
Obviously the complexity of such an algorithm depends on the complexity of Step 3 of this procedure, where the  $d(X_a,X_b)$ is evaluated.

A direct attack on evaluating  $d(X_a,X_b)$ might seem to take time with complexity $O(k!)$ and hence not be polynomial in $n$.  This is because the distance between $X_a= \{x_1,\dots, x_k \}$ and  $X_b=\{y_1, \dots, y_k\} $ with respect to $G \sympow k$ is the sum of the distances with respect to $G$ between $x_j$ and $y_{\pi(j)}$, minimized over all permutations $\pi$ that permute $k$ symbols. There are then $k!$ possible permutations and $k$ distances to sum for each instance.  This however is not the case, since the problem of evaluating $d(X_a,X_b)$ is actually equivalent to the minimum assignment problem, which can be solved in $O(k^3)$ time using the celebrated Kuhn-Munkres algorithm \cite[Page 52]{schrijver2004combinatorial}, after one first computes all pairwise distances in $G$.

We now explain how combinatorial optimization algorithms from graph theory can be used to compute lower bounds on $d_{j,k}$ can be evaluated in polynomial time. 
\begin{enumerate}
\item Algorithm \ref{alg:dist} computes the all pairwise distances in $G$. This is achieved using breath-first-search on every vertex. Since  breadth-first search on any vertex produces a shortest path tree  in linear time \cite[Theorem 6.4]{schrijver2004combinatorial}, and there are $n$ such vertices,  Algorithm \ref{alg:dist} runs in $O(n^2)$ time.
\item Algorithm~\ref{alg:distXY} evaluates distances between given vertices in $G \sympow k$. It turns out that the evaluation of $d(X_a,X_b)$ is equivalent to the well-known minimum assignment problem in the field of combinatorial optimization. First, evaluate $Z=X_a \cap X_b$ and set $X = X_a \setminus Z$ and $Y = X_b \setminus Z$.  Consider a complete bipartite graph with every vertex in $x\in X$ is connected to a vertex in $y \in Y$ by a weighted edge. The weight of the edge $\{x,y\}$ in the bipartite graph is equal to the distance between $x$ and $y$ given by $d(x,y)$. The problem of computing $d(X_a,X_b)$ is then equivalent to finding the perfect matching (set of edges such that every vertex belongs to exactly one edge) on this bipartite graph, such that the sum of the weights on these matchings is minimized. But this is precisely equal to the minimum assignment problem, which can be solved using the Kuhn-Munkres algorithm. We therefore just need to generate the cost matrix for the minimum assignment problem in this algorithm to utilize the Kuhn-Munkres algorithm.
\end{enumerate}

We would be able to easily compute the generalized diameter of $G \sympow k$ exactly, if we only knew how to optimally select $j+1$ of its vertices in $G \sympow k$. 
Without such knowledge, we can use  Algorithm~\ref{alg:select} to randomly select $j+1$ vertices in $G \sympow k$.

We completely describe our algorithm to compute upper bounds on the eigenvalues of $L_k$ in Algorithm~\ref{alg:upper-bounds}.
\begin{algo}
{\texttt{Upp}$(j,k,G=(V,E))$, Upper bounds on $\lambda_j(L_k)$. \label{alg:upper-bounds}}
\begin{algorithmic}[0]

\State  \textbf{Initialization}
\State $1 \le k \le \frac n 2$.
\State $1 \le j < \binom n k $.
\State $\beta \gets$ maximum vertex degree of $G$.
\State $\mu \gets 2k \beta$ \Comment upper bound on $\lambda_{\rm max}(L_k)$
\State $D \gets \texttt{Dist}(G)$ \Comment From Algorithm \ref{alg:dist}
\State \textbf{end initialization}
\newline

\State $(X_1, \dots, X_{j+1}) \gets \texttt{SEL}_k(j,V)$ \Comment From Algorithm \ref{alg:select}
\State $d \gets \infty$
\ForAll {$a,b = 1,\dots,j+1 : a < b$}  
	\State $d_{X_a,X_b} \gets \texttt{Dist}(X_a,X_b,D)$ \Comment From Algorithm \ref{alg:distXY}
    \If {$d_{X_a,X_b}<d$}
    	\State $d \gets d_{X_a,X_b}$
    \EndIf
\EndFor
\newline
\If{$d\ge 2$}
	\State $u \gets \mu (1-2/(1+\binom n k ^{1/d}))$
\Else
	\State $u \gets \infty$
\EndIf 
\State \textbf{return} $u$
\end{algorithmic}
\end{algo}
Since there are $\binom{j+1}{2}$ possible pairwise distances amongst $X_1,\dots ,X_{j+1}$ that we must consider, 
the time complexity of running Algorithm~\ref{alg:upper-bounds} is 
\begin{align}
O(n^2)+O(j^2 k^3).
\end{align}
This thereby leads to an algorithm that evaluates a lower bound for $d_{j,k}$ in time polynomial in $n$, $j$ and $k$. 
This then leads to our formal result, which we give in the following theorem.
\begin{theorem}
\label{theorem:upper-bound-algo}
Let $G = (V,E)$ be any graph with $n$ vertices. Let $2 \le k \le n/2$ and $1\le j \le \binom n k - 1$. Then Algorithm \ref{alg:upper-bounds} can compute an upper bound on $\lambda_j(L_k)$ in $O(n^2)+O(j^2 k^3)$ time.
\end{theorem}
Thus for all $k$ and $j$ polynomial in $n$, upper bounds on the eigenvalues of the ferromagnetic HH can be computed in time polynomial in $n$.
Such an algorithm would outperform a direct solver for Laplacians \cite{spielman2014nearly} whenever $k \ge 3$.

\section{Lower bounds for the Heisenberg spectrum}
\label{sec:computing-lowerbounds}
A property of graphs that we focus on are their associated isoperimetric inequalities.
These isoperimetric inequalities on graphs allow us to define the notion of the isoperimetric dimension of a graph.
Now let $X$ be a set of vertices and $\partial X$ be its boundary.
In this case, the edge boundary of $X$ is just the set of edges in $E$ with exactly one vertex in $X$ and one vertex in $V \setminus X$.
Then the edge-isoperimetric inequality on graphs \cite{Alon1986} 
is any lower bound of the form
\begin{align}
|\partial X| \ge c |X|^{1-1/d}
\end{align}
that holds for every vertex subset $X$ of size at most half the cardinality of $V$.
The utility of these isoperimetric inequalities in the case of continuous manifolds lies in their applicability for example to give bounds on the principal frequency of a vibrating membrane \cite{payne1967isoperimetric}. 
The rationale behind seeking edge-isoperimetric inequalities for the graphs $G$ 
lies in the fact that such inequalities can yield spectral bounds on the eigenvalues of the normalized Laplacians of $G$ \cite{ChY95}, and hence also of the Laplacians.
Since the Heisenberg Hamiltonian is just a direct sum of Laplacians of $G \sympow k$, edge-isoperimetric inequalities on $G \sympow k$ can then yield bounds on the corresponding energy eigenvalues of the Heisenberg Hamiltonian. 

In this section, we prove several technical results relating to the edge-isoperimetric inequalities on the symmetric products of graphs.
Roughly speaking, our results allow us to establish the isoperimetric properties of $G \sympow k$ in terms of the isoperimetric properties of certain subgraphs of the graph $G$.
In particular, these subgraphs are vertex induced subgraphs of $G$ where a number of vertices and their corresponding edges are deleted from $G$. 
Our technical result applies to graphs with a finite number of vertices.
In Theorem \ref{thm:symprod-seminorm}, we prove that that if deleting any $k-1$ vertices from a finite graph $G$ yields a vertex induced subgraph that has a dimension $\delta$ with isoperimetric number $C$,
then a lower bound on the size of the edge-boundary of a subset of vertices $\Omega$ in $G \sympow k$ is given in terms of the size of the edge-boundary of $\Omega$ in the Johnson graph that is the $k$-th symmetric product of the complete graph. 

The proof relies crucially on the fact that the size of an edge boundary of a set $X$ can be written as a Sobolev seminorm of the indicator function of $X$. 
This implies that edge-isoperimetric inequalities can be written in terms of 
the Sobolev seminorm of an indicator function and an appropriate functional of that indicator function, 
as we shall see in Section \ref{sec:sobolev}.
Also, we use Tillich's observation of a one-to-one correspondence between edge-isoperimetric inequalities and inequalities relating the Sobolev seminorm of functions and an appropriate functional evaluated on those functions \cite{Til00}.
Together, these insights allow us to obtain lower bounds on the size of the edge-boundary of the subsets of vertices in $G \sympow k$. 

\subsection{Sobolev inequalities on graphs}
\label{sec:sobolev}
Recall that an edge-isoperimetric inequality for a graph $G=(V,E)$ has the form
\begin{align}
|\partial X| \ge C |X|^{1-1/d}, \quad \forall X\subseteq V: |X|= k,
\end{align}
where $k=1,\dots , |V|/2$. The point of this section is that the size of the edge-boundary $|\partial X|$  can be written in terms of a discrete Sobolev seminorm, and this allows us to obtain some interesting insights.
Namely, given a graph $G=(V,E)$ and a function $f : V \to \mathbb R$ on the vertex set, the discrete Sobolev seminorm of $f$
corresponding to the edge set $E$ is defined by
\begin{align}
\| f \|_{E} = \sum_{ \{ u,v \} \in E } |f(u) - f(v)| . \notag
\end{align}
Now consider the case where $f={\bf 1}_X$ where 
${\bf 1}_X : V \to \{0,1\}$ is an indicator function on $X$ so that for all $X \subseteq V$, ${\bf 1}_X(x) = 1$ if $x \in X$ and ${\bf 1}_X(x) = 0$ if $x \in V \setminus X$.
Then it is clear that 
\begin{align}
|\partial X| = \| {\bf 1}_{X} \|_E. 
\end{align}
We call any inequality which involves the Sobolev seminorm $\|\cdot \|_E$, such as the one above, a discrete Sobolev inequality. 

The analytic inequalities of Tillich \cite[Theorem 2]{Til00} establish the equivalence between edge-isoperimetric inequalities and discrete Sobolev inequalities on functionals 
that map functions from $\Phi_V$ to non-negative real numbers,
where $\Phi_V$ denotes the set of all functions $f : V \to \mathbb R$.
To state Tillich's theorem succinctly, we introduce the following definition.
\begin{definition}\label{def:1}
Given $C>0$ and a functional $\rho : \Phi_V \to \mathbb R^+$, we say that $G$ is $(C,\rho)$-isoperimetric if for every $X \subseteq V$, we have $\| {\bf 1}_X \|_E \ge C \rho( {\bf 1}_X ).$
\end{definition}
By not requiring that $|X|\le |V|/2$, an implicit constraint on the choice of feasible functionals $\rho$ that can satisfy the discrete Sobolev inequality in Definition \ref{def:1} is imposed.  

 We state Tillich's result on functionals that are also seminorms in the following theorem.
\begin{theorem}[{\cite[Theorem 2]{Til00}}]
\label{thm:tillich}
Let $G =(V,E)$ be a graph, $C > 0$, and $\rho$ be a seminorm on $\Phi_V$.
Then $G$ is $(C,\rho)$-isoperimetric if and only if $\|f \|_E \ge C \rho( f )$ for every function $f : V \to \mathbb R$. 
\end{theorem}
Imposing the additional constraint $|X|\le V/2$ would allow ourselves to work with a larger family of seminorms $\rho$, 
but Theorem \ref{thm:tillich} would need appropriate modification, which we do not address in this paper.
Working without the constraint $|X|\le V/2$ allowed Tillich to derive edge-isoperimetric inequalities for graphs with a countably infinite number of vertices.

In this article, we restrict our attention to the functionals $g_p$ and $\rho_p$ for $p \ge 1$, where
\begin{align} 
 g_p(f)  &= \left(  \frac{1}{|V|} \sum_{x,y \in V } |f(x) - f(y)|^p  \right)^{1/p}, \\
\rho_p( f )    &= \left( \sum_{x \in V} |f(x) - \mathbb E(f) | ^{p} \right) ^{1/p} , 
\end{align}
where 
\begin{align}
\mathbb E(f) =  \frac{1}{|V| }\sum_{v \in V } f(v)
\end{align}
denotes the expectation value of $f$.
It is then easy to show that 
\begin{align} 
g_p({\bf 1}_X)  &= \left( \frac{2 |X| |V \setminus X|}{|V|}\right)^{1/p}, \label{eq:gpform}\\
\rho_p (   {\bf 1}_X ) &=
 \left( \sum_{x \in V} \left| {\bf 1}_X(x) - \frac{|X|}{|V|} \right| ^{p} \right) ^{1/p}.
\end{align}
Note that when $g_p$ and $\rho_p$ are evaluated on ${\bf 1}_X$, 
they are invariant under the substitution of $X$ with $V \setminus X$.

 The discrete Sobolev inequality is closely related to the isoperimetric number and dimension of a graph as given in the following proposition, which is obvious from definitions. 
\begin{proposition}
Let $G=(V,E)$ be graph and $C>0$ and $\delta>1$. Then the following are true.
\begin{enumerate} 
\item If $V$ is finite and $G$ is $(C,g_{\delta/(\delta-1)})$-isoperimetric, then $G$ has an dimension of $\delta$ with isoperimetric number $C$.
\item If $V$ is finite and $G$ has dimension $\delta$ with isoperimetric number $C$, then $G$ is $(2^{-\delta/(\delta-1)}C,g_{\delta/(\delta-1)})$-isoperimetric.
\end{enumerate}
\end{proposition}
Hence we can address finite-sized graphs with the functionals $\rho_p$ using the two-sided bounds on $\rho_p({\bf 1}_X)$ in terms of $g_p({\bf 1}_X)$ as given in the following lemma.
Note that when $p=1$, we get $\rho_1 (   {\bf 1}_X ) = g_1({\bf 1}_X)$ for any vertex subset $X$.
\begin{lemma}
\label{lem:rho_p}
Let $G = (V,E) $ be a graph, $X \subseteq V$ and $p \ge 1$. Then
\begin{align}
\frac{1}{2^{1-1/p}} g_p({\bf 1}_X)
\le 
\rho_p (   {\bf 1}_X ) 
\le 
g_p({\bf 1}_X)
. \notag
\end{align}
\end{lemma}

\begin{proof}
By definition, 
$\rho_p (   {\bf 1}_X ) =
 \left( \sum_{x \in V} \left| {\bf 1}_X(x) - \frac{|X|}{|V|} \right| ^{p} \right) ^{1/p}$. 
Splitting the summation over $V$ into the disjoint subsets $X$ and $V \setminus X$ yields
\begin{align}
\rho_p (   {\bf 1}_X )
&=  
 \left( 
 	|X| \left( 1 - \frac{|X|}{|V|} \right) ^p 
 +
   (|V| - |X| ) \left( \frac{|X|}{|V|}\right)^p
 \right) ^{1/p}.
\end{align}
Since $ \left( 1 - \frac{|X|}{|V|} \right) ^p \le  \left( 1 - \frac{|X|}{|V|} \right)$ and $ \left( \frac{|X|}{|V|}\right)^p \le  \left( \frac{|X|}{|V|}\right)$ for $p\ge 1$, we get $\rho_p (   {\bf 1}_X ) \le g_p({\bf 1}_X)$.
Since both 
  $ |X| \left( 1 - \frac{|X|}{|V|} \right) ^p  $ and 
  $ |V \setminus X| \left( \frac{|X|}{|V|}\right)^p$ 
  are at least $ \left( \frac{|X| |V \setminus X| }{|V|}\right)
  (\frac 1 2 )^{p-1},$
  we get $\rho_{p}({\bf 1}_{X}) \ge g_p({\bf 1}_X) (1/2^{p-1})^{1/p}$.
\end{proof}
We remark that Lemma \ref{lem:rho_p} is tight when $p=1$, because then we would have
\begin{align}
 g_1({\bf 1}_X)
\le 
\rho_1 (   {\bf 1}_X ) 
\le 
g_1({\bf 1}_X),
\end{align}
which implies that $\rho_1 (   {\bf 1}_X )  = g_1 (   {\bf 1}_X ) $. 
The scenario $p=1$ occurs for graphs with infinite dimensions, and expander graphs are examples of such graphs.

Lemma \ref{lem:rho_p} implies the following for $C>0$ and $\delta > 1$.
\begin{enumerate}
\item If a graph is $(C, \rho_{\delta/(\delta-1)})$-isoperimetric, 
the graph also has dimension $\delta$ with isoperimetric number $2^{-\delta} C $.
\item If a graph has dimension $\delta$ with isoperimetric number $C$,
the graph is also $(2^{-\delta/(\delta-1)}C,g_{\delta/(\delta-1)})$-isoperimetric.
\end{enumerate}
In what follows, we use Theorem \ref{thm:tillich} where $\rho =\rho_p$ for $p \ge 1$.   

\subsection{The symmetric product of finite graphs}
Now we address the edge-isoperimetric problem on the graph $G^{\{k\}}$ when $G$ has a finite number of vertices, for a fixed positive integer $k = 2, \dots, \lfloor |V|/ 2 \rfloor$. 
Again we rely on the edge-isoperimetric properties of the vertex-induced subgraphs of a graph $G$. 
A key ingredient of our proof is a bijection between sets, described by the following proposition.
\begin{proposition}
\label{prop:same-cardinality}
Let $V$ be a countable set and $k$ be a integer such that $k = 1,\dots, |V|$.
Then the sets 
$\mathcal A = \{ (W ;x) : W \subseteq V , |W| = k-1 , x \in V \setminus W\}$ 
and
$\mathcal A' = \{ (X ; x) : X \subseteq V , |X| = k , x \in X\}$
have the same cardinality.
\end{proposition}
\begin{proof}
Let $f : \mathcal A \to \mathcal A' $ where $f \mapsto (W ; x)  =  (W \cup \{x\} ; x)$
for all $W \subseteq V$ and $x \in V \setminus W$. 
The map $f$ is invertible, and is therefore a bijection from $\mathcal A$ to $\mathcal A'$. 
Hence $\mathcal A$ and $\mathcal A'$ have the same cardinality. 
\end{proof}

 We obtain here a lower bound on $|\partial \Omega|$, which is the size of the edge boundary of any vertex subset $\Omega$ in $G^{\{k\}}$.
 Our lower bound on $|\partial \Omega|$ is provided in terms of
$|\partial_J\Omega|$, which is the size of the edge boundary of $\Omega$ in the Johnson graph $J(n,k)$.
\begin{theorem}
\label{thm:symprod-seminorm}
Let $G=(V,E)$ be a graph with $n$ vertices, and let $p \ge 1$ and $C >0$.
Suppose that every vertex-induced subgraph of $G$ with $n-k+1$ vertices is $(C,\rho_p)$-isoperimetric.
Then for every $\Omega \subseteq V^{\{k\}}$,
\begin{align}
  |\partial \Omega|  \ge \frac{C}{n-k+1} (2 |\partial_J \Omega |)^{1/p}. \notag
\end{align}
\end{theorem}
Note that the inequality in Theorem \ref{thm:symprod-seminorm} is tight for $p=1$. 
To see this, let us consider a trivial scenario where $G$ is the complete graph on $n$ vertices, and $k=1$.
For the complete graph, we can compute the edge boundary of any vertex subset $X$ exactly. 
Denoting $x=|X|$ and $n=|V|$, we have $|\partial X| = \min\{x,n-x\} (n-1)$.
recall that from (\ref{eq:gpform}) that $g_1({\bf 1}_X) = \frac{2x (n-x)}{n}$. 
Then the edge-isoperimetric inequality for the complete graph with respect to the seminorm $g_1$ is equivalent to 
\begin{align}
(n-1 )  \min\{x,n-x\} \ge C 2x(1-x/n).
\end{align}
This inequality holds trivially when $x=0$, so let us consider $x \ge 1$.
Now focus on the scenario where $x \le n/2$. 
Then $(n-1 ) x \ge C 2x(1-x/n)$, which is equivalent to 
$(n-1 )  \ge C 2(1-x/n)$ and $C \le \frac{ n-1  }{2(1-x/n)}$.
To minimum upper bound for $C$ in this case is attained for $x=1$, and thus we have
$C \le \frac n 2$.
Now consider the scenario where $\frac n 2 <x \le n$. When $x=n$, the inequality again holds trivially. so we consider 
$\frac n 2 < x \le n-1$.
Then the inequality we are faced with is $(n-1 )  y \ge C 2y(1-y/n),$
where $y=n-x$. Since we just finished analyzing this scenario, we can conclude thatt the optimal isoperimetric constant is $C=\frac n 2$ for the complete graph.
Substituting this example into Theorem \ref{thm:symprod-seminorm}, since $G \sympow 1 = G$, we get for the complete graph
\begin{align}
|\partial X| \ge \frac{n}{2} \frac{1}{n} (2 |\partial X| )
\end{align}
which is equivalent to $1 \ge 1$ and hence the inequality in Theorem \ref{thm:symprod-seminorm} is tight for the complete graph.

\begin{proof}[Proof of Theorem \ref{thm:symprod-seminorm}]
For all $\Omega \subseteq V^{\{k\}}$,
note that $|\partial \Omega| = \|  {\bf 1}_\Omega \|_{E^{\{k\}}}$.
Two $k$-sets $X$ and $Y$ in $\Omega$ are adjacent in the graph $G^{\{k\}}$ if and only if the symmetric difference of $X$ and $Y$ is an edge in $E$.
Hence 
\begin{align}
|\partial \Omega|
&= 
\sum_{\substack{ W \subset V  \\ |W| = k-1} }
\sum_{ \{ u,v\} \in E[V \setminus W] }
\left| {\bf 1}_\Omega(W \cup \{ u \})  - { \bf 1 }_{\Omega}( W \cup \{v\} ) \right|.
\end{align}
Applying Theorem \ref{thm:tillich} with seminorm $\rho_p$ on each induced subgraph $G[V \setminus W]$ for every $(k-1)$-set $W$ with respect to the function ${\bf 1}_\Omega( W \cup \{\cdot \} )$, we get
\begin{align}
|\partial \Omega|
\ge 
\sum_{\substack{ W \subset V  \\ |W| = k-1} }
C \left( 
	\sum_{ x  \in V \setminus W }
\left| {\bf 1}_\Omega(W \cup \{ x \})  
- \sum_{y \in V \setminus W}
\frac{ { \bf 1 }_{\Omega}( W \cup \{y\}) }{ n- k +1 }  
\right|^p \right)^{1/p}. \label{eq:proof-sobolev1}
\end{align}
By subadditivity of the function $( \cdot)^{1/p}$ for all $p \ge 1$,
the inequality (\ref{eq:proof-sobolev1}) becomes
\begin{align}
|\partial \Omega|
\ge 
C \left( 
\sum_{\substack{ W \subset V  \\ |W| = k-1} }
	\sum_{ x \in V \setminus W }
\left| {\bf 1}_\Omega(W \cup \{ x \})  
- \sum_{y \in V \setminus W}
\frac{ { \bf 1 }_{\Omega}( W \cup \{y\}) }{ n- k +1 }  
\right|^p \right)^{1/p}. \label{eq:proof-sobolev1.1}
\end{align}
By Proposition \ref{prop:same-cardinality} we can reorder the summation in (\ref{eq:proof-sobolev1.1}) to get
\begin{align}
|\partial \Omega|
&\ge 
C \left( \sum_{X \in V^{\{k\}}} \sum_{x \in X} 
\left| 
	{\bf 1}_\Omega(X) - \sum_{ y \in V \setminus (X \setminus \{x\} )  } \frac{{\bf 1}_\Omega  (X \triangle \{x, y\} )  }{n-k+1}
\right|^p \right)^{1/p} .   \label{eq:proof-sobolev2}
\end{align}
Each $k$-set $X$ appearing in the inequality (\ref{eq:proof-sobolev2}) either belongs to $\Omega$ or not.
Applying simple arithmetic on the right hand side of (\ref{eq:proof-sobolev2}) above then yields
\begin{align}
&  
C \left( \sum_{X \in \Omega} \sum_{x \in X} 
\left(   \sum_{ y \in V \setminus (X \setminus \{x\} )  } \frac{1-{\bf 1}_\Omega  (X \triangle \{x, y\} )  }{n-k+1} \right)^p
+
 \sum_{X \notin \Omega} \sum_{x \in X}  
	\left( \sum_{ y \in V \setminus (X \setminus \{x\} )  } \frac{{\bf 1}_\Omega  (X \triangle \{x, y\} )  }{n-k+1} \right)^p
	\right)^{1/p} . \label{eq:proof-sobolev3}
\end{align}
Using the inequality $\left(\sum_{i} x_i \right)^p \ge \sum_{i} x_i^p$ for non-negative $x_i$, the expression (\ref{eq:proof-sobolev3}) becomes
\begin{align}
&
C \left( \sum_{X \in \Omega} \sum_{x \in X} 
   \sum_{ y \in V \setminus (X \setminus \{x\} )  } \frac{1-{\bf 1}_\Omega  (X \triangle \{x, y\} )  }{\left(n-k+1 \right)^p} 
+
  \sum_{X \notin \Omega} \sum_{x \in X}  
	\sum_{ y \in V \setminus (X \setminus \{x\} )  } \frac{{\bf 1}_\Omega  (X \triangle \{x, y\} )  }{\left( n-k+1\right)^p} 
	\right)^{1/p}  
	\notag\\
&=
 \frac{ C } { n-k+1 }  \left( 2 \sum_{X \notin \Omega} \sum_{x \in X}  
	 \sum_{ y \in V \setminus (X \setminus \{x\} )  } {\bf 1}_\Omega  (X \triangle \{x, y\} )  
\right)^{1/p}. \notag
\end{align}
To complete the proof, note that
\begin{align}
  \sum_{X \notin \Omega} \sum_{x \in X}  
	 \sum_{ y \in V \setminus (X \setminus \{x\} )  }  {\bf 1}_\Omega  (X \triangle \{x, y\} )   = |\partial_J \Omega| . \notag
\end{align}
\end{proof}
The eigenvalues of the combinatorial Laplacian of the Johnson graph $J(n,k)$ for $k=0,\dots , \lfloor n/2 \rfloor$ are 
$j(n+1-j)$ with multiplicities $\binom n j - \binom n {j-1}$, 
where $j = 0,\dots, k$ \cite[Section 12.3.2]{brouwer2011spectra}.
If $\lambda$ is the second smallest eigenvalue of the combinatorial Laplacian of a graph, then that graph is $(\frac{\lambda}{2}, g_1)$-isoperimetric \cite[Lemma 13.7.1]{godsil2001algebraic}.
Since the second smallest eigenvalue of the combinatorial Laplacian of the Johnson graph $J(n,k)$ is always $n$, $ |\partial_J \Omega | \ge 
  \frac{n}{2} g_1({\bf 1}_\Omega) $
for every $\Omega \subseteq V^{\{k\}}$.
Hence 
\begin{align}
(2 |\partial_J \Omega| )^{1/p} \ge 
(n g_1({\bf 1}_\Omega) )^{1/p} =n^{1/p} g_p({\bf 1}_\Omega)  . \label{eq:johnson-eip}
\end{align}
Using (\ref{eq:johnson-eip}) with Theorem \ref{thm:symprod-seminorm} together with Lemma \ref{lem:rho_p} yields the following corollary. 
\begin{corollary}
\label{coro:1}
Let $G=(V,E)$ be a graph with $n$ vertices, and let  $p \ge 1$ and $C >0$.  
Suppose that every vertex-induced subgraph of $G$ with $n-k+1$ vertices is $(C,\rho_p)$-isoperimetric.
Then $G^{\{k\}}$ is $(\frac{  C n^{1/p} }{n-k+1}, g_p)$-isoperimetric and $(\frac{  C n^{1/p} }{n-k+1}, \rho_p)$-isoperimetric.
\end{corollary}
This corollary plays a central role in the next subsection.

\subsection{Lower bounds from isoperimetric considerations}

If one were to compute the eigenvalues of $L_k$ directly, one may quickly run into computational difficulties. The reason is twofold. 
First, the size of the matrix $L_k$ is $\binom n k$, and in general scales exponentially with $n$. 
This leads to the difficulty in evaluating the eigenvalues of $L_k$ when one does not desire to utilize a computer with both exponential memory that runs in  exponential time.
In view of this problem, our methodology of obtain lower bounds on the eigenvalues of $L_k$ will be handy.
 The algorithms to compute lower bounds that we introduce from graph theory will considerably outperform algorithms that directly compute the eigenvalues of $L_k$.
 Instead of studying the symmetric products $G \sympow k$, we restrict our attention to the vertex-induced subgraphs of $G$.

 When one deletes vertices from a graph $G$ along with the corresponding edges, one obtains a vertex-induced subgraph of $G$.
 We denote the set of all graphs obtained by deleting exactly $k-1$ vertices from $G$ as $\mathcal V(G,k)$. 
 Clearly, there are $\binom n {k-1}$ graphs in the set $\mathcal V(G,k)$.
 From Corollary \ref{coro:1}, we know that if $C$ is less than the isoperimetric number of every graph in $\mathcal V(G,k)$ with corresponding dimension $\delta$, then the graph $G \sympow k$ has isoperimetric dimension $\delta$ with isoperimetric number at least 
 \begin{align}
 C \frac{n^{1-1/\delta}}{n-k+1}.
 \end{align}

We now proceed to outline how lower bounds on the eigenvalues of $L_k$ can be obtained from geometric considerations the graphs $G\sympow k$.
To achieve this, we will first illustrate how lower bounds on the spectrum of a graph Laplacians can depend only on the graph's geometry.
We begin by introducing some notation.
Let $D_G = \sum_{v \in V} d_v |v\>\<v|$ denote the degree matrix of a graph $G=(V,E)$. 
Let $A_G$ denote the adjacency matrix of a graph, which means that it is a matrix with matrix elements equal to either 0 or 1, and where  $\<u|A_G|v\>= 1$ iff the vertex $u$ is adjacent to $v$.
Let $L_G$ denote the Laplacian of a graph, which can be written as $D_G-A_G$.
In this subsection, we have the following theorem, which is essentially a Chung-Yau type bound \cite{ChY95} with Ostrovskii's correction \cite{Ost05} for unnormalized Laplacians.
\begin{theorem} \label{theorem:lowerbound}
Let a graph $G=(V,E)$ have dimension $\delta > 2$ with isoperimetric number $c$.
Let $b$ and $\beta$ be the minimum and maximum vertex degrees of $G$  respectively.
Then  
\begin{align}
\lambda_j(L_G) \ge 
\frac{bc^2}{16e \beta^{2}} 
 \left(
\frac{\delta-2}{\delta-1}
\right)^2
\left(
\frac { j \beta} {18|E|}
\right)^{2/\delta}.
\label{eq:lowerbound}
\end{align} 
\end{theorem}
When a graph is connected, its degree matrix is non-singular, 
and we can write its normalized Laplacian of $G$ as 
\begin{align}
\widetilde {L}_G= D_G^{-1/2} {L}_G D_G^{-1/2}
\end{align}
The proof of Theorem \ref{theorem:lowerbound} relies trivially on the result on the corresponding result for lower bounds on the spectrum of normalized Laplacians. 
The connection is given by the following lemma.
\begin{lemma} \label{lem:Lbounds}
If the graph has minimum and maximum vertex degrees given by $b$ and $\beta$ respectively,   
 \begin{align}
b \lambda_j(\widetilde L_G)   \le \lambda_j(L_G) \le \beta \lambda_j(\widetilde L_G). \label{eq:Laplacians}
\end{align} 
\end{lemma}
\begin{proof}
Denoting the $i$-th largest singular value of a matrix $A$ of size $d_a$ as $s_i(A)$ with $s_1(A) \ge \dots \ge s_{d_a}(A)$, 
we have from Ref \cite[Problem III.6.5]{Bhatia} the inequalities 
\begin{align}
 s_i(AB) \le s_i(A) s_1(B), \quad s_i(AB) \le s_1(A) s_i(B) .
\end{align}
Applying the above inequalities iteratively,  
it follows that
\begin{align}
s_i(\widetilde {L}_G) 
&= s_i( D_G^{-1/2} {L}_G D_G^{-1/2} ) \le s_i(L_G) s_1( D_G^{-1 } ), \notag\\
s_i( {L}_G)  
&= s_i( D_G^{ 1/2} {\widetilde L}_G D_G^{ 1/2} ) \le s_i(\widetilde L_G) s_1( D_G ). \label{eq:Lbound} 
\end{align}
Since the matrices $D_G, D_G^{-1}, L_G$ and $\widetilde L_G$ are positive semidefinite,
their singular values are equivalent to their eigenvalues. 
The largest eigenvalue of $D_G$ and $D_G^{-1}$ are $\beta$ and $b^{-1}$ respectively. Hence the inequalities (\ref{eq:Lbound}) then give the result.
\end{proof}
Lower bounds on the eigenvalues of the normalized Laplacian can be obtained from the graph's Sobolev inequalities,
as shown in the seminal work of Chung and Yau \cite{ChY95}. 
Because of a gap in the proof in \cite{ChY95} as shown by Ostrovskii \cite[after Equation 8]{Ost05}, we have to take Ostrovskii's correction into account when we prove the corresponding lower bounds on the graph's Laplacian which we state explicitly in Theorem \ref{theorem:lowerbound}.
\begin{proof}[Proof of Theorem \ref{theorem:lowerbound}]
For a graph $G=(V,E)$, denote the volume of a subset of vertices $X$ as
 ${\rm vol}(X) = \sum_{v \in X}d_v$. 
Also let ${\rm vol}(G) = \sum_{v \in V}d_v$ denote the sum of all vertex degrees in the graph $G$. 
The isoperimetric inequality we focus on is
\begin{align}
|\partial X| \ge c_\delta ({\rm vol}(X))^{1-1/\delta}
\end{align}
where ${\rm vol}(X) \le {\rm vol} (V \setminus X)$.
Note here that ${\rm vol}(X)$ is in general different from the number of vertices in $X$. 
While $|X|$ counts the number of vertices in $X$, the volume ${\rm vol}(X)$ counts the sum of all vertex degrees of vertices in $X$.
We may also interpret ${\rm vol}(X)$ as the number of vertices in $X$ multiplied by the average degree of the vertices in $X$.
The Sobolev inequality on graphs has the form 
\begin{align}
\sum_{\{u, v\} \in E} |f(u) - f(v) |^2
\ge
A \min_{\mu \in \mathbb R} 
\left(
\sum_{v \in V} |f(v)-\mu|^{\alpha} d_v 
\right)^{2/\alpha}, \label{eq:mysob1}
\end{align}
where $\alpha = \frac{2 \delta}{\delta-2}$.
Typically $A$ depends on $c_\delta$ and $\delta$.
Chung and Yau proved when the above Sobolev inequality holds for a graph, the eigenvalues of the graph's normalized Laplacians satisfy the lower bound
\begin{align}
\lambda_j(\widetilde L_G) \ge \frac{A}{e 3^{4/\delta} } \left(
j / {\rm vol}(G)
\right)^{2/\delta}.
\end{align}
When $\delta >2$, the inequality (\ref{eq:mysob1}) holds with $A = \frac{c_\delta^2 }{16} \left(
\frac{\delta-2}{\delta-1}
\right)^2$
using Ostrovskii's Sobolev inequality \cite[(8)]{Ost05}.
Using this fact with Lemma \ref{lem:Lbounds}, we get
\begin{align}
\lambda_j(L_G)
\ge
 b  
 \frac{c_\delta^2 }{16e} \left(
\frac{\delta-2}{\delta-1}
\right)^2  \left(
{j \over 9{\rm vol}(G)}
\right)^{2/\delta}. 
\end{align}
It remains to relate $c_\delta$ to $c$.
Let $\beta$ be the maximum vertex degree of $G$.
Since $G$ has isoperimetric dimension $\delta$ and isoperimetric number $c$,
its vertex subsets $X$ satisfy the bound
\begin{align}
|\partial X| 
&\ge c \min\{|X|, |V| - |X| \}^{1-1/\delta}  \ge \frac{c}{\beta ^{1-1/\delta}} \min \left\{ { \rm vol}({X}) , {\rm vol}({V \setminus X}) \right\} ^{1-1/\delta} .
\end{align}
Hence we can take $c_\delta=c/\beta^{1-1/\delta}$.
The hand-shaking Lemma also implies that ${\rm vol}(G) =2|E|$, and we get the result.
\end{proof}

Using Theorem~\ref{theorem:lowerbound}, we can easily obtain lower bounds on the eigenvalues of $L_k$ using $b_k$ and $\beta_k$, which are the minimum and maximum vertex degrees of $G \sympow k$ respectively. 
Note that $\beta_1$ denotes the maximum number of interacting neighbors each spin experiences in the Heisenberg ferromanget.
To bound $b_k$ and $\beta_k$, note that every vertex in $G \sympow k$ is a set of vertices in $G$ with $k$ elements. 
Therefore the vertex degree of $\{x_1,\dots, x_k\}$ in $G \sympow k$ is just the edge-boundary of $\{x_1,\dots, x_k\}$ in $G$.
Thus $b_{k} \ge c k^{1-1/\delta}$ 
whenever $G$ has dimension $\delta$ with isoperimetric number $c$.
Also, when $\beta$ is the maximum vertex degree of $G$, we trivially have 
and $\beta_{k} \le k \beta_1 $.
Hence Corollary \ref{coro:1} implies that $c_k\ge a_k \frac{n^{1-1/\delta_k}}{n-k+1}$, 
where every vertex-induced subgraph of $G$ with $k-1$ deleted vertices 
has dimension $\delta_k$ with isoperimetric number $a_k$. 
The number of edges in $G \sympow k$ is at most $\beta_k \binom n k / 2$ where $n$ is the number of spins.
Then if $\delta_k >2$ for $k=1,\dots, n/2$, Theorem~\ref{theorem:lowerbound} implies that
\begin{align}
\lambda_j(L_k)
&\ge
\frac{c k^{-1/\delta} a_k^2}{16e k \beta_1^2} 
\left(
\frac{n^{1-1/\delta_k}}{n-k+1}
\right)^2
 \left(
\frac{\delta-2}{\delta-1}
\right)^2
\left(
\frac j {9   \binom n k }
\right)^{2/\delta_k}.
\label{eq:lowerbound}
\end{align}
To numerically estimate $a_k$, it suffices to numerically compute the isoperimetric numbers of graphs $K \in \mathcal V(G,k)$ with vertex set $V(K)$ and edge set $E(K)$.
To find the isoperimetric number of $K$, 
we need to solve its corresponding edge-isoperimetric problem (EIP) on $K$, 
which involves finding 
\begin{align}
\min \{ |\partial X | : X \subseteq V(K), |X| = j \}
\end{align}
for every $1 \le  j \le |V(K)|/2$.
While solving the EIP exactly is NP-hard \cite{garey1976some,brandes2009vertex},
we conjecture that there can be approximation algorithms to approximately solve the EIP in polynomial time. 
\begin{conjecture}\label{conjecture-EIP}
Let $G=(V,E)$ be a graph. 
For every $k=1,\dots, |v|/2$, let $e_k = \min\{|\partial X|: X \subseteq V, |X|=k \} $. Then for every $\epsilon >0$ and and for every $k=1,\dots, |V|/2$, there exists a polynomial time approximation algorithm that computes $e'_k$ such that $(1-\epsilon) e_k \le e'_k \le e_k$.
\end{conjecture}
A reason why Conjecture \ref{conjecture-EIP} might be true is because for a multitude of different NP-hard problems, there do exist approximation algorithms that have efficient runtimes \cite{hochbaum1996approximation}.
If our Conjecture \ref{conjecture-EIP} holds, 
then lower bounds on the eigenvalues can be evaluated in 
$O({\rm poly}(n) n^{k-1})$ time with $O(n)$ memory.
In contrast, computing the eigenvalues of $L_k$ directly in practice requires a computer in $O(n^{3k})$ time and $O(n^{2k})$ memory.
Even using the best asymptotic algorithm for matrix multiplication would require at least $O(n^{2k})$ time \cite{demmel2007fast} and $O(n^{2k})$ memory.

\section{Discussions}
\label{sec:conclusion} 
In this paper, we obtain many bounds on the spectrum of the ferromagnetic HHs. 
For this, we rely on tools from graph theory and matrix analysis. 
Obviously, with these bounds on the eigenvalues of the Heisenberg ferromagnet, 
one can easily compute bounds on thermodynamic quantities of the corresponding Heisenberg models such as free energy. 

With regards to upper bounds based on graph distances, there remains a potential to further tighten our bounds by optimizing over the partitions used in Eq.~(4.22) of Corollary 4.4 in Ref.~\cite{CGY96}. 
This is however beyond the scope of the current paper and we leave this for future investigation.
With regards to the lower bounds based on isoperimetric inequalities, 
we wish to point out that the edge-isoperimetric problem for the Johnson graph, also known as the problem of Kleitman and West \cite{Har91}, remains unsolved. 
Given this fact, better edge-isoperimetric inequalities for the Johnson graph will improve the edge-isoperimetric inequalities of the symmetric product of finite graphs given in Corollary \ref{coro:1}.
Also advances in the theory of the graph expansion properties of vertex induced subgraphs will certainly also improve the bounds given in this corollary.
Directly deriving lower bounds on the combinatorial Laplacian of a graph from discrete Sobolev inequalities can also help to improve the constants involved in the bound.
Moreover, a polynomial-time approximation algorithm for solving the edge-isoperimetric problem for graphs (Conjecture \ref{conjecture-EIP}) would together with the methods already in this paper, yield a polynomial-time algorithm for computing lower bounds for the eigenvalues of the ferromagnetic HH.
 
 To recap, in the spin half case, the computational basis of the ferromagnetic HM can be represented by a binary string. 
Each binary string is represented as a vertex, and interactions represented as edges between the vertices.
In the spin-half case, each exchange interaction is equivalent to a swap operator, and acts as a transposition on the binary strings. The relationship between different binary strings under transpositions that correspond to the interaction are represented as a graph. 
Because transpositions leave the Hamming weight of these binary strings invariant, the HH naturally decomposes into a direct sum of graphs labeled by all the possible Hamming weights from 0 to $n$.

One might wonder how the results here could generalize to the spin $S$ case.
We briefly sketch how one might proceed to achieve this.
We can observe that the computational basis of the ferromagnetic HM can be represented by a $(2S+1)$-nary string.
We can represent these $(2S+1)$-nary strings as vertices on a graph, and interactions as relationships between the vertices.
In this representation, the spin-S exchange operator maps a $(2S+1)$-nary string to a linear combination of $(2S+1)$-nary strings. Since one can show that the coefficients of this linear combination are non-negative, if all non-zero exchange constants are the same, the coefficients can rescale to allow us to interpret them as probabilities of transitions from one vertex to another vertex.
Since the spin-$S$ exchange operator conserves total spin, the $(2S+1)$-nary strings naturally partition into disjoint subsets, where only strings in different partitions do not interact, and strings in the same partition can have their interactions represented as a Markov model. 
We expect the spectrum HH to thereby be related to the spectrum of the associated $(2S+1)$ Markov models. 
Markov models describe stochastic transitions between a set of discrete states and are well-studied.
We therefore expect that connections between the theory of Markov models and spin-$S$ HMs can bring similar insights into bound the spectrum of spin-$S$ HMs.

\section{Acknowledgements}
YO likes to thank Robert Seiringer and anonymous referees for their comments and recommendations that have helped to improve this manuscript.
YO acknowledges support from the Singapore National Research Foundation under NRF Award NRF-NRFF2013-01, the U.S. Air Force Office of Scientific Research under AOARD grant FA2386-18-1-4003,
and the Singapore Ministry of Education. 
This work was supported by the EPSRC (grant no. EP/M024261/1)

\bibliographystyle{alpha} 
\newcommand{\etalchar}[1]{$^{#1}$}

\end{document}